\RequirePackage{mathtools}
\RequirePackage{thm-restate}
\documentclass[runningheads]{llncs}

\usepackage{graphicx}
\usepackage{wrapfig}
\usepackage{amsfonts}
\usepackage{amsmath}
\usepackage{amsthm}
\usepackage{xspace}
\usepackage{hyperref}
\usepackage{cite}
\usepackage{xcolor}
\graphicspath{{./figures/}}%helpful if your graphic files are in another directory

\newcommand{\mypar}[1]{\smallskip\noindent{\bfseries\boldmath#1}}

\newcommand{\goal}{\ensuremath{t}\xspace}
\newcommand{\peak}[0]{\wedge}
\newcommand{\strategypath}[0]{\ensuremath{P}}
\newcommand{\projectedpath}[0]{\ensuremath{P^*}}

\newcommand{\vis}[1]{\ensuremath{\mathit{Vis(#1)}}}

\newcommand{\connectedcomponent}[0]{\ensuremath{\mathcal{C}}}

\newcommand{\turningpoint}{\ensuremath{p_{\rotatebox[origin=c]{180}{$\hookrightarrow$}}}}
\newcommand{\pointopt}{\ensuremath{p}}

\newcommand{\subpath}{subpath\xspace}
\newcommand{\subpaths}{subpaths\xspace}
\newcommand{\turnsegment}{\ensuremath{h_\ell^{\rotatebox[origin=c]{180}{$\hookrightarrow$}}}}

\begin{document}

\title{Competitive Searching over Terrains
\thanks{This research was initiated at the 6th Workshop on Applied Geometric Algorithms (AGA 2023), Otterlo, The Netherlands, April 7--21, 2023.}
} %TODO Please add

%\titlerunning{Dummy short title} %TODO optional, please use if title is longer than one line

% LNCS AUTHOR MACRO
\author{Sarita de Berg\inst{1} \and Nathan van Beusekom\inst{2} \and Max van Mulken\inst{2} \and Kevin Verbeek\inst{2} \and Jules Wulms\inst{2}}

\authorrunning{S. de Berg, N. van Beusekom, M. van Mulken, K. Verbeek, and J. Wulms} %TODO mandatory. First: Use abbreviated first/middle names. Second (only in severe cases): Use first author plus 'et al.'

\institute{Utrecht University, The Netherlands\\ \email{S.deBerg@uu.nl} 
\and TU Eindhoven, The Netherlands\\ \email{[n.a.c.v.beusekom|m.j.m.v.mulken|k.a.b.verbeek|j.j.h.m.wulms]@tue.nl}}

\maketitle

\begin{abstract}
We study a variant of the searching problem where the environment consists of a known terrain and the goal is to obtain visibility of an unknown target point on the surface of the terrain. The searcher starts on the surface of the terrain and is allowed to fly above the terrain. The goal is to devise a searching strategy that minimizes the \emph{competitive ratio}, that is, the worst-case ratio between the distance traveled by the searching strategy and the minimum travel distance needed to detect the target. For $1.5$D terrains we show that any searching strategy has a competitive ratio of at least $\sqrt{82}$ and we present a nearly-optimal searching strategy that achieves a competitive ratio of $3\sqrt{19/2} \approx \sqrt{82} + 0.19$. This strategy extends directly to the case where the searcher has no knowledge of the terrain beforehand. For $2.5$D terrains we show that the optimal competitive ratio depends on the maximum slope $\lambda$ of the terrain, and is hence unbounded in general. Specifically, we provide a lower bound on the competitive ratio of $\Omega(\sqrt{\lambda})$. Finally, we complement the lower bound with a searching strategy based on the maximum slope of the known terrain, which achieves a competitive ratio of $O(\sqrt{\lambda})$. 
\end{abstract}

\section{Introduction}
The development of autonomous mobile systems has garnered a lot of attention recently. With self-driving cars and autonomous path-finding robots becoming more commonplace, the demand for efficient algorithms to govern the decision-making of these systems has risen as well. A class of problems that naturally arises from these developments is the class of \emph{searching problems}, also known as \emph{searching games}: given an environment, move through the environment to find a target at an unknown location. Many variants of this general problem have been studied in literature, typically differing in the type of search environment, the way the searcher can move through the environment, and the way the target can be detected. In this paper we consider a variant of the problem that is motivated by searching terrains using a flying (autonomous) drone with mounted cameras/sensors, as for example in search-and-rescue operations. Specifically, our environment is defined by a height function $T_d: \mathbb{R}^d \rightarrow \mathbb{R}$. For $d=1$ we refer to the terrain as a $1.5$D terrain, and for $d=2$ we refer to the terrain as a $2.5$D terrain. We omit the dimension $d$ from the terrain function $T_d$ when it is clear from the context. The terrain is known to the searcher, and the searcher can fly anywhere above the terrain. The target $t$ is discovered if it can be seen by the searcher along a straight line.
The goal is to devise a \emph{searching strategy} (that is, a search path) that finds the (unknown) target $t$ as quickly as possible (see Figure~\ref{fig:1D-strategy}). To the best of our knowledge, this natural variant of the searching problem has not been studied before.

As is common for searching problems, we analyze the quality of the searching strategy using \emph{competitive analysis}. For that we consider the ratio between the travel distance using our searching strategy and the minimum travel distance needed to detect that target. The maximum value of this ratio over all possible environments and all possible target locations is the \emph{competitive ratio}~$c$ of the searching strategy. The goal is to find a searching strategy that minimizes~$c$.

\begin{figure}[t]
    \centering
    \includegraphics[page=2]{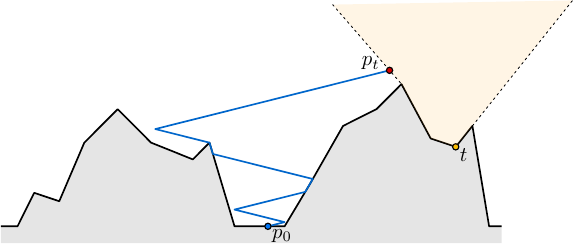}
    \caption{Our searching strategy starts from $p_0$ and then follows the blue searching path. When the searcher reaches $p_\goal$, it can see the target $t$.}
    \label{fig:1D-strategy}
\end{figure}

\mypar{Related work.}
Searching problems or searching games have been studied extensively in the past decades. Here, we mostly restrict ourselves to searching problems with a geometric environment and continuous motion. One of the most fundamental searching problems is the problem of searching on an infinite line, where the target is detected only when the searcher passes over it on the line. An optimal strategy for this problem was discovered by Beck and Newman~\cite{beck1970yet} and works as follows. Assuming that the distance to the target is at least one, we first move one to the right from the starting point. Next, we move back to the starting point and then move two to the left. We then repeat this process, alternating between moving to the right and left of the starting point, every time doubling the distance from the starting point. This searching strategy has a competitive ratio of $9$, which is optimal for this problem.

In subsequent work, researchers have studied searching problems for many other different environments, including lines and grids~\cite{DBLP:journals/iandc/Baeza-YatesCR93}, line arrangements~\cite{DBLP:conf/isaac/BoutsCGKM18}, and graphs~\cite{DBLP:journals/ijcga/BoseCDT17, DBLP:journals/tcs/BoseM04}. Other variants include searching on a line when an upper bound on the distance to $\goal$ is known~\cite{DBLP:journals/tcs/BoseCD15, DBLP:journals/dam/HipkeIKL99}, when turns contribute to the cost of the solution~\cite{demaine2006online}, or when there are multiple searchers~\cite{DBLP:journals/comgeo/Baeza-YatesS95}. 

In settings where the environment is $2$-dimensional (or higher), it is not possible to visit every point in the environment, and hence we need to consider different ways of detecting the target. In these settings, the target is often considered detected if it can be seen directly from the searcher's position along a straight, unobstructed line. One example is the problem of finding a target point $\goal$ inside a simple polygon $P$ with $n$ vertices~\cite{klein1997algorithmische}. For this problem the optimal competitive ratio is unbounded, as it necessarily depends on $n$~\cite{Kleinberg94,schuierer1999line}. As a result, researchers have explored this searching problem for special sub-classes of polygons where a constant competitive ratio can be achieved. A polygon $P$ is considered a \emph{street} if there exist two vertices $s$ and $t$ on its border such that the two boundary chains leading from $s$ to $t$ are mutually weakly visible. Searching for an unknown point in a street can be done with a competitive ratio of $\sqrt{2}$, which is optimal~\cite{icking2004optimal,klein1992walking}. There is a large body of further work on searching problems in variants of streets~\cite{datta1995competitive, wei2019walking,lopez1996generalized}, star-shaped polygons~\cite{icking1995searching, lopez2003searching}, or among obstacles~\cite{blum1997, kalyanasundaram1993competitive}. For a comprehensive overview of these variants, see~\cite{ghosh2010online}. Specifically, L\'opez-Ortiz and Schuierer~\cite{lopez2003searching} obtain a competitive ratio of~11.51 for star-shaped polygons; note that $1.5$D terrains are a special type of unbounded star-shaped polygons. 

Other problems strongly related to searching problems are the \emph{exploration problems}, for which the goal is to move through the interior of an unknown environment to gain visibility of its entire interior. Here, the competitive ratio relates the length of the searching strategy to the shortest watchman tour. An unknown simple polygon can be fully explored with competitive ratio $26.5$~\cite{hoffmann2001polygon}. For polygons with holes, the competitive ratio is dependent on the number of holes~\cite{deng1998learn, georges2013online}, whereas in a rectilinear polygon without holes a competitive ratio as low as $3/2$ can be achieved~\cite{hammar2006competitive}.  Complementary to this problem is the exploration of the \emph{outer} boundary of a simple polygon, where a 23.78 or 26.5 competitive ratio can be achieved for a convex or concave polygon, respectively~\cite{Wei21polygonboundary}.

\mypar{Contributions.}
Given a starting position $p_0$ on the surface of the terrain (we assume without loss of generality that $p_0$ is at the origin and that $T_d(p_0) = 0$), the goal is to devise an efficient searching strategy to find an unknown target point $\goal$ on the surface of the terrain, where the searcher can detect $\goal$ if it is visible from the searcher's position 
along a straight unobstructed line (see Figure~\ref{fig:1D-strategy}). In our problem the searcher is not restricted to the surface of the terrain, but it is allowed to move to any position on or above the terrain. 

In Section~\ref{sec:1.5D} we consider the problem for $1.5$D terrains. We first prove that any searching strategy for this problem must have a competitive ratio of at least $\sqrt{82}$. We then present a searching strategy with a competitive ratio of $3\sqrt{19/2} \approx \sqrt{82} + 0.19$. Our searching strategy is a combination of the classic searching strategy on an infinite line with additional vertical movement. 

In Section~\ref{sec:2.5D} we consider the problem for $2.5$D terrains. Here we show that no searching strategy can achieve a bounded competitive ratio, as the competitive ratio necessarily depends on the maximum slope, or Lipschitz constant, $\lambda$ of the terrain function. Specifically, we show that the competitive ratio of any searching strategy for this problem is at least $\Omega(\sqrt{\lambda})$. We then present a novel searching strategy that achieves a competitive ratio of $O(\sqrt{\lambda})$, which is thus asymptotically optimal in terms of $\lambda$. 

In our searching problems we assume that the terrain~$T_d$ is known to the searcher. In Section~\ref{sec:conclusion} we conclude that our strategy for a $1.5$D terrain is directly applicable if this is not the case, and discuss to what degree the results for the $2.5$D case extend as well.

\section{Competitive searching on 1.5D terrains}\label{sec:1.5D}
In this section we consider our searching problem on a 1.5D terrain, defined by the height function~$T_1$. For 1.5D terrains, the visibility region~$\vis{t}$ of a target $t$ can be defined as the set of all points $q$ for which the line segment $tq$ does not properly intersect the terrain. That is, the region that contains all points that can see $t$. Since we can assume that $p_0$ does not see $t$, there is a half-line originating from $t$ that needs to be crossed to enter $\vis{t}$ (see Figure~\ref{fig:1D-strategy}). We call this half-line a visibility ray. Thus, the goal of any searching strategy on 1.5D terrain is to enter~$\vis{\goal}$ by crossing a visibility ray $r$ originating at the target point~$\goal$. We first establish a lower bound on the competitive ratio of any searching strategy.

\begin{figure}[b]
    \centering
    \includegraphics{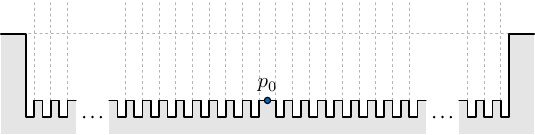}
    \caption{The lower bound construction for the 1.5D case.}
    \label{fig:1.5D-LB}
\end{figure}

\begin{theorem}\label{thm:1D-lower-bound}
    The competitive ratio for searching on 1.5D terrains is at least~$\sqrt{82}$. 
\end{theorem}
\begin{proof}%\sarita{rewrote this proof a bit to make it more clear when we set~$h$}
    Consider a terrain with pits at integer x-coordinates that have infinite slope, and thus cast (near-)vertical visibility rays. Furthermore, there is a rectangular mountain infinitely far away casting a horizontal visibility ray at some height $h$ (see Figure~\ref{fig:1.5D-LB}). Because of the lower bound of $9$ on the competitive ratio for searching on a line~\cite{beck1970yet}, there must be a pit at distance $d$ such that the distance covered by the optimal strategy for searching on a line is $9d$. We set~$h = d$.

    Now consider a search path $P$ for this terrain. If $P$ does not reach height $h$ after covering a horizontal distance of $9 d$, then $P$ travels a distance of more than $\sqrt{81 d^2 + h^2} = \sqrt{82} d$ before reaching the horizontal ray at height $h$, resulting in a competitive ratio of more than $\sqrt{82}$. Otherwise, $P$ travels a distance of at least $\sqrt{81 d^2 + h^2} = \sqrt{82} d$ before reaching the pit at distance $d$. Hence the competitive ratio of $P$ is at least $\sqrt{82}$.
\end{proof}

\mypar{Searching strategy.}
Our strategy is based on the classic searching strategy on an infinite line with additional vertical movement. Specifically, we construct a \emph{projected path}~$\projectedpath$ that acts as a guide for the actual \emph{searching path}~$\strategypath$. In the description of our strategy we make use of infinitesimally small steps at the start, as this simplifies the description and analysis. We later mention how to avoid this and make our strategy feasible in practice, under the assumption that the length of the shortest path to the target is bounded from below by some positive constant. Starting from~$p_0$, $\projectedpath$ moves diagonally with slope~$s$ to one side, for a horizontal distance of $\varepsilon$, and then moves back to $x=0$, again with slope~$s$. Subsequently, $\projectedpath$ alternates between moving left and right of~$p_0$, doubling the horizontal distance when moving away from~$p_0$, and always using slope~$s$. 
As a result, $\projectedpath$~consists of $xy$-monotone \emph{segments}, and \emph{turning points} where the $x$-direction swaps.
Specifically, $\projectedpath$ is defined by the following functions for $i\in\mathbb{Z}$:
\begin{align*}
h^i_r(x) = s\cdot(2^i + x) &\phantom{=}\text{for } i \text{ odd and } -2^{i-2}\leq x \leq 2^{i-1}, \\
h^i_\ell(x) = s\cdot(2^i - x) &\phantom{=}\text{for }i \text{ even and } -2^{i-1}\leq x \leq 2^{i-2} .
\end{align*}

We call the line segments $h^i_r$ the \emph{right} segments of~$\projectedpath$, and the line segments $h^i_\ell$ the \emph{left} segments of~$\projectedpath$. Observe that for two consecutive segments, the values at the ends of the domains coincide, which results in~$\projectedpath$ being a connected path. Specifically, we get $h^i_r(2^{i-1}) = h^{i+1}_\ell(2^{(i+1)-2})$ and $h^{i-1}_\ell(-2^{(i-1)-1}) = h^i_r(-2^{i-2})$. 

The actual search path~$\strategypath$ follows $\projectedpath$ (see Figure~\ref{fig:notation-1D}). Whenever $\strategypath$ hits the terrain, it follows the terrain upwards until it can continue moving diagonally with slope~$s$ again. This diagonal part of~$\strategypath$ does not coincide with~$\projectedpath$, so once $\strategypath$ hits~$\projectedpath$, $\strategypath$ starts following~$\projectedpath$ again. Observe that $\strategypath$ still consists of $xy$-monotone polygonal chains and turning points, albeit both can differ from $\projectedpath$. We refer to the monotone chains of~$P$ as right and left \subpaths of~$P$, when they are monotone in the positive and negative $x$-direction, respectively. We also refer to a line segment in such a right or left \subpath as a left or right segment.
We will choose $s$ later to optimize the resulting competitive ratio.

\mypar{Preliminaries and definitions.}
The target $\goal$ can be seen from any point in the visibility region $\vis{\goal}$.
We consider all possible visibility rays~$r(s_r, d_r)$ that can separate $\vis{\goal}$ and~$p_0$, where $s_r$ is the slope of~$r$ in the positive $x$-direction, and $d_r$ is the distance between $r$ and~$p_0$.

\begin{figure}[t]
    \centering
    \includegraphics[page=1]{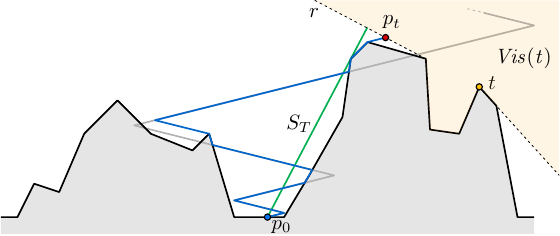}
    \caption{The strategy $\strategypath$ (blue) and the projected path $\projectedpath$ (gray). The goal is seen at~$p_\goal$.}
    \label{fig:notation-1D}
\end{figure}

Let $S$ be the line segment between~$p_0$ and $r$ that is perpendicular to~$r$, so $|S| = d_r$. Furthermore, let $S_T$ be the shortest geodesic path from $p_0$ to $r$, taking the terrain $T$ into account. If $S$ does not properly intersect~$T$, then $|S_T|=d_r$. Note that the last line segment of $S_T$ is perpendicular to~$r$.
Finally, let $p_\goal \in P$ be the point where $\strategypath$ crosses~$r$ to enter~$\vis{t}$. 

We define $\tau(r)$ as the distance traversed over $P$ until $r$ is crossed, i.e. from~$p_0$ until $p_\goal$, and $c(r) = \frac{\tau(r)}{|S_T|}$ as the competitive ratio to cross a ray~$r$. To simplify our proofs, we additionally introduce the following definitions.
Let $p_\goal = (x,y)$, then we define $\tau^*(r) = y \sqrt{1+s^2}/s$. So, $\tau^*(r)$ is the length of a path with slope~$s$ up to height~$y$.
When $\strategypath$ deviates from~$\projectedpath$, the projected path is intersected by~$T$ and hence $\strategypath$ is steeper than $\projectedpath$. It follows that $\tau(r)\leq\tau^*(r)$. We use the ratio $c^*(r) =\frac{\tau^*(r)}{|S_T|}$ in our proofs, and analyze the maximum of $c^*(r)$ over all instances to bound the competitive ratio of our searching strategy from above.

For computing the competitive ratio, we only need to consider visibility rays originating from one side. This is due to the symmetric nature of our strategy: we can take any instance with a visibility ray originating left of $p_0$, and transform it into a case equivalent to having the visibility ray originating from the right of $p_0$. We achieve this by scaling all distances in $\projectedpath$ by 2 to get the horizontally symmetrical path. To see this observe that for $-2^{i-2}\leq x \leq 2^{i-1}$ and $x' = -x$
\begin{equation*}
    h^i_r(x) = s\cdot(2^i + x) = 2\cdot s\cdot (2^{i-1} + x/2) = 2 \cdot h^{i-1}_\ell(x'/2).
\end{equation*}
Additionally, we mirror $T$ horizontally in $p_0$, hence $r$ is also mirrored with respect to~$p_0$. We thus consider only visibility rays that originate to the right of~$p_0$.

Finally, since $T$ is a height function, the visibility region of any point above the terrain includes the vertical ray cast upwards from that point. Thus, for the visibility ray~$r$ that separates $\vis{\goal}$ from $p_0$ it holds that $s_r \leq 0$.

\mypar{Proof structure.}
To determine the competitive ratio of our searching strategy, we analyze the competitive ratio in a worst-case instance $(T, r)$ for $c^*(r)$, where $T$ is the terrain and $r = r(s_r, d_r)$ is the visibility ray from the target. To that end we first establish several properties that must hold in some worst-case instance. To exclude various instances from consideration, we can use the following lower bound on the competitive ratio of our strategy. A competitive ratio below this bound would contradict the lower bound for searching on a line~\cite{beck1970yet}.
\begin{restatable}{lemma}{lowerboundS}\label{lem:lowerboundS}
        The competitive ratio $c^*(r)$ is at least $9 \sqrt{1+s^2}$.
\end{restatable} 

\begin{proof}
        Suppose the competitive ratio of $c^*$ would be strictly smaller than $9\sqrt{1+s^2}$. Consider the 1-dimensional searching problem of searching for a target point on a line. For this problem, we know the competitive ratio is lower bounded by 9~\cite{beck1970yet}. If we project the path of our strategy onto this line, we obtain a strategy to search for a point on the line. If our strategy has $\tau^*(r) \leq d$, the horizontal distance traveled is at most $d\cos(\arctan(s)) = d/\sqrt{1+s^2}$. Thus, the competitive ratio for searching on a line would be strictly smaller than $9\sqrt{1+s^2}/\sqrt{1+s^2} = 9$, which is a contradiction.
\end{proof}

In the remainder of this section we show that a worst-case instance $(T, r(s_r, d_r))$ has the following properties:
\begin{itemize}
    \item The ray $r$ is arbitrarily close to a turning point $\turningpoint$ of $P$ when $p_t$ lies on a right \subpath (Lemma~\ref{lem:1D-move-r}).
    \item The ray $r$ satisfies $s_r \leq -s$ and $p_t$ lies on a right \subpath (Lemma~\ref{lem:flat-rays}).
    \item If $\turningpoint$ is not a turning point of $\projectedpath$, and thus a local maximum $\peak$ of $T$ intersects $\projectedpath$ before $\turningpoint$, then $\peak$ is at $\turningpoint$ and $r$ is vertical (Lemma~\ref{lem:1D-peak-and-vertical}).
    \item If $\turningpoint$ coincides with a turning point of $\projectedpath$, then $r$ is vertical (Lemma~\ref{lem:1D-unobstructed}).
\end{itemize}

After establishing these properties, the remaining cases can easily be analyzed directly in the proof of Theorem~\ref{thm:1D-comp-ratio}. 

\mypar{Close to turning point.} We first prove two lemmata that help us establish the properties indicated above. Lemma~\ref{lem:growth-bounds} follows from the quotient rule of derivatives.

\begin{restatable}{lemma}{growthbounds}\label{lem:growth-bounds}
    Let $f, g, h: \mathbb{R}_{>0} \rightarrow \mathbb{R}_{>0}$ be differentiable functions such that $f(x) = \frac{g(x)}{h(x)}$ and $\frac{\text{\emph{d}} g}{\text{\emph{d}}x}, \frac{\text{\emph{d}} h}{\text{\emph{d}}x} > 0$ for any $x > 0$. Then, $\frac{\text{\emph{d}} f}{\text{\emph{d}} x} < 0$ if and only if $\frac{\text{\emph{d}} g}{\text{\emph{d}} x} / \frac{\text{\emph{d}} h}{\text{\emph{d}} x} < f(x)$. 
\end{restatable}

\begin{proof}
    Given $f(x) = \frac{g(x)}{h(x)}$, we get
    \begin{equation*}
        \frac{\text{d}f}{\text{d} x} = \frac{\frac{\text{d} g}{\text{d} x}h(x) - g(x)\frac{\text{d} h}{\text{d} x}}{h(x)^2}
    \end{equation*}
    
    Since $h(x)^2 > 0$, this gives us
    \begin{align*}
        \frac{\text{d} f}{\text{d} x} &< 0 &\iff\\
        \frac{\text{d} g}{\text{d} x}h(x) - g(x)\frac{\text{d} h}{\text{d} x} &< 0 &\iff\\
        \frac{\text{d} g}{\text{d} x}h(x) &< g(x)\frac{\text{d} h}{\text{d} x} &\iff\\
        \frac{\text{d} g}{\text{d} x}/\frac{\text{d} h}{\text{d} x} &< \frac{g(x)}{h(x)} &\iff\\
        \frac{\text{d} g}{\text{d} x}/\frac{\text{d} h}{\text{d} x} &< f(x). & \qedhere
    \end{align*}
\end{proof}

\begin{lemma}\label{lem:strip-ratio}
    Let $(T,r(s_r, d_r))$ be an instance where $p_\goal$ lies on a right \subpath of $\strategypath$ with slope $s$, and let $c^*(d_r) = \frac{\tau^* (r)}{|S_T|}$. If $\frac{1}{9} < s \leq1$, then $\frac{\text{\emph{d}} c^*}{\text{\emph{d}} d_r} < 0$. 
\end{lemma}
\begin{proof} 
    By Lemma~\ref{lem:growth-bounds}, to prove $\frac{\text{d} c^*}{\text{d} d_r} < 0$ it is sufficient to show that $\frac{\text{d} \tau^*(r)}{\text{d} d_r}/\frac{\text{d}|S_T|}{\text{d} d_r} < c^*(d_r)$. Since $\frac{\text{d}|S_T|}{\text{d} d_r} \geq 1$, we get that $\frac{\text{d} \tau^*(r)}{\text{d} d_r}/\frac{\text{d}|S_T|}{\text{d} d_r} \leq \frac{\text{d} \tau^*(r)}{\text{d} d_r}$. Consider decreasing $d_r$, i.e., moving the ray $r$ towards the origin. If $s \leq 1$, we can bound the ratio between the change in $\tau^*(r)$ and the change in~$d_r$ as follows.
    \[\frac{\text{d} \tau^*(r)}{\text{d} d_r} \leq \frac{\sqrt{s^2 + 1}}{s} < 9\sqrt{1+s^2}\leq c^*(d_r)\]

    The second step holds for $s > \frac{1}{9}$, and the final step follows from Lemma~\ref{lem:lowerboundS}.
\end{proof}

Lemma~\ref{lem:strip-ratio} implies that, as long as $p_\goal$ lies on a right \subpath of $\strategypath$, decreasing $d_r$ increases~$c^*(r)$. 

\begin{restatable}{lemma}{oneDmover}\label{lem:1D-move-r}
    %Let $r(s_r, d_r)$ be the target visibility ray such that the lowest point where this ray intersects $\strategypath$ lies on a right segment of $\strategypath$ with slope $s$. If $\frac{1}{9} < s \leq 1$ and for fixed $s_r$, in the worst case we can minimize~$d_r$ without decreasing the competitive ratio, until the intersection of the target ray and $\strategypath$ is infinitesimally close to a turning point of $\strategypath$.
    In a worst-case instance~$(T, r(s_r,d_r))$ where $p_t$ lies on a right \subpath of~$\strategypath$, if $\frac{1}{9} < s \leq 1$, then $r$ is infinitesimally close to a turning point of~$\strategypath$.
\end{restatable}

\begin{proof}
    Assume for contradiction that the lowest intersection of~$r$ and $\strategypath$ lies on a right \subpath of~$\strategypath$, but $r$ is further than $\varepsilon>0$ away from a turning point of~$\strategypath$.
    
    Given ray $r(s_r, d_r)$, decreasing $d_r$ gives a worse competitive ratio as long as $r$ intersects $\strategypath$ on a right segment of $\strategypath$ with slope $s$ according to Lemma~\ref{lem:strip-ratio}, contradicting that~$(T,r)$ is a worst case. By construction, as soon as we can no longer move the ray towards the origin while keeping its intersection point on a right segment of $\strategypath$ with slope~$s$, we must either encounter a turning point of $\strategypath$, or the slope of~$\strategypath$ is not equal to~$s$.

    \begin{figure}
        \centering
        \includegraphics{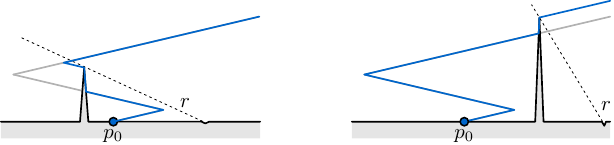}
        \caption{By moving $r$ to the left until we are longer on a right segment with slope $s$, we can either hit a left \subpath (case on the left), or a right segment with slope greater than $s$ (case on the right).}
        \label{fig:close-to-turning}
    \end{figure}

    In the former case, the lemma is proven, so assume that the latter case applies. When the slope of~$\strategypath$ is not~$s$, $\strategypath$ must route along~$T$, and hence~$r$ would encounter a local maximum~$\peak$ of~$T$. We distinguish two cases (see Figure~\ref{fig:close-to-turning}): either a left \subpath of~$\strategypath$ is routed over~$\peak$, or a right \subpath of~$\strategypath$ is routed over~$\peak$. When a left \subpath of~$\strategypath$ is routed over~$\peak$, lowering~$\peak$ to the highest intersected segment of $\projectedpath$ possibly decreases $|S_T|$, while $\tau^*(r)$ remains the same. Hence, the competitive ratio~$c^*(r)$ has not decreased either. Next, Lemma~\ref{lem:strip-ratio} can once again be applied, resulting in a contradiction.
    
    When a right \subpath of~$\strategypath$ is routed over~$\peak$, then~$\peak$ must be located on the right side of~$p_0$, otherwise a turning point of~$\strategypath$ is located at~$\peak$. Let~$y$ be the height of $\peak$. This means that $|S_T|\geq y$ and $\tau^*(r) = \frac{y\cdot\sqrt{1+s^2}}{s}$, so we get a competitive ratio of
    \begin{equation*}
        \frac{\tau^*(r)}{|S_T|}\leq \frac{y\sqrt{1+s^2}}{ys} = \frac{\sqrt{1+s^2}}{s} < 9\sqrt{1+s^2}
    \end{equation*}
    In the last step we use that $s>\frac{1}{9}$. As this is below the lower bound of Lemma~\ref{lem:lowerboundS} this contradicts that~$(T,r)$ is a worst-case instance.
\end{proof}

\mypar{Flat visibility rays.} Next, we deal with all visibility rays~$r(s_r,d_r)$ for which $s_r > -s$, which we call \emph{flat} visibility rays. All other visibility rays, which have a slope of at most~$-s$, we define as \emph{steep} visibility rays. 
We show that~$r$ is never flat in a worst-case instance, and $p_t$ must then lie on a right \subpath of~\strategypath.

\begin{restatable}{lemma}{flatrays}\label{lem:flat-rays}
    In a worst-case instance~$(T, r(s_r,d_r))$, if $\frac{2}{9} < s \leq 1$ then $s_r \leq -s$ and $p_t$ lies on a right \subpath of \strategypath.
\end{restatable}

\begin{proof}
    First, assume that $s_r > -s$, we distinguish two cases depending on whether $p_\goal$ lies on a right or left \subpath of~$P$.

    When~$p_\goal$ lies on a right \subpath of~$\strategypath$, Lemma~\ref{lem:1D-move-r} implies that the visibility ray~$r$ is infinitesimally close to a turning point. Because $s_r > -s$, it must be that $r$ passes infinitesimally close to the upper turning point of a left \subpath of~$\strategypath$. In this case, Lemma~\ref{lem:strip-ratio} can even be applied until $p_\goal$ lies on the upper turning point of the left \subpath, resulting in the next case.

    \begin{figure}
        \centering
        \includegraphics{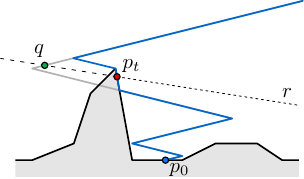}
        \caption{A flat visibility ray $r$ that intersects a left subpath.}
        \label{fig:lemma-flat}
    \end{figure}
    If $p_\goal$ lies on a left \subpath but does not lie on~$\projectedpath$, then consider the intersection point $q$ of the extension of $r$ with $\projectedpath$, see Figure~\ref{fig:lemma-flat}. Observe that $q$ is at least as high as $p_\goal$. We can therefore consider the instance where $q=p_\goal$, as this can only decrease $|S_T|$ and increase $\tau^*(r)$. If $q$ lies on a right segment of~$\projectedpath$, then  $p_\goal$ is moved to this right segment as well, which contradicts the assumption that $p_\goal$ lies on a left \subpath. Hence, in the worst case $p_\goal$ lies on a left segment.
    
    Now that we know that in the worst case~$p_\goal$ lies somewhere on the $i$th left segment, we can derive the following bounds on~$\tau^*(r)$ and~$|S_T|$. We use the length 
    of~$\projectedpath$ up to the upper end of~$h^i_\ell$ as an upper bound on~$\tau^*(r)$.
    Thus~$\tau^*(r)\leq 3\cdot2^{i-1}\cdot \sqrt{1+s^2}$. Similarly, as~$r$ passes above the lower end of~$h^{i}_\ell$, the height of the lower end of~$h^{i}_\ell$ is a lower bound on~$|S_T|$. So,~$|S_T|\geq 3s\cdot 2^{i-2}$.
    Flat visibility rays hence result in 
    \begin{equation*}
        c^*(r) = \frac{\tau^*(r)}{|S_T|} \leq \frac{3\cdot2^{i-1}\cdot \sqrt{1+s^2}}{3s\cdot 2^{i-2}}\leq \frac{2 \sqrt{1+s^2}}{s} < 9\sqrt{1+s^2}.
    \end{equation*}

    The latter inequality holds when $s>\frac{2}{9}$. As this competitive ratio is below the lower bound of Lemma~\ref{lem:lowerboundS}, this contradicts that~$(T,r)$ is a worst case instance. It follows that in a worst-case instance $s_r \leq -s$.

    To conclude the proof we consider the case where $p_t$ lies on a left \subpath of~$\strategypath$ and $s_r \leq -s$. Observe that this is possible only when $T$ intersects $\projectedpath$ on the left side of~$p_0$ (similar to Figure~\ref{fig:lemma-flat}). Again consider the intersection point $q$ of $r$ with $\projectedpath$, and observe that $q$ must lie on a right segment of~$\projectedpath$ since $s_r \leq -s$. We can then apply the following argument from earlier:~$q$ is higher than $p_\goal$. Thus the instance where $q=p_\goal$ leads to a worse competitive ratio, since $|S_T|$ is unaffected by removing the obstruction of~$\projectedpath$ left of~$p_0$. This contradicts that~$(T,r)$ is a worst case instance.
\end{proof}

Lemma~\ref{lem:flat-rays} shows that for an upper bound on the competitive ratio we do not have to consider flat visibility rays~$r$.

From now on, we thus consider only steep visibility rays with~$p_t$ on a right \subpath. Let~$\turningpoint$ denote the turning point infinitesimally close to~$r$. 
Note that~$\turningpoint$ must lie on
%$h_\ell^{i-1}$
the final \emph{left} segment of $\projectedpath$ that is on the search path. We denote this segment by $\turnsegment$.

\mypar{Obstructed search path.}
For steep visibility rays, first consider the case where~$\turningpoint$ is not a turning point of $\projectedpath$, i.e. $T$ obstructs the right segment before~$\turnsegment$ (see Figure~\ref{fig:highest-ratio}).
We call a local maximum of $T$ a peak, denoted by~$\peak$. We prove that a worst-case instance $(T,r)$ has the following three properties.
\begin{figure}
    \centering
    \includegraphics{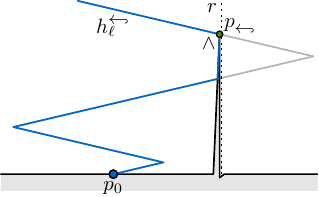}
    \caption{The search path is obstructed and $r$ is steep.}
    \label{fig:highest-ratio}
\end{figure}
\begin{enumerate}
    \item If a peak lies on~$\projectedpath$ at~$\turningpoint$, then the visibility ray~$r$ is vertical (Lemma~\ref{lem:vertical-r});
    \item if the visibility ray~$r$ is vertical, then a peak lies on~$\projectedpath$ at~$\turningpoint$ (Lemma~\ref{lem:1D-move-peak});
    \item either a peak lies on~$\projectedpath$ at~$\turningpoint$, or the visibility ray~$r$ is vertical (Lemma~\ref{lem:1D-peak-or-vertical}).
\end{enumerate}

The following lemma follows directly from the above statements.

\begin{restatable}{lemma}{oneDpeakvertical}\label{lem:1D-peak-and-vertical}
    Let~$(T, r(s_r,d_r))$ be a worst-case instance where~$\turningpoint$ is not a turning point of $\projectedpath$.
    If $\frac{2}{9} < s < \frac{4}{9}$, then a peak lies on~$\turnsegment$ at~$\turningpoint$, and $s_r=-\infty$.
\end{restatable}

\begin{proof}
    We can apply Lemma~\ref{lem:1D-peak-or-vertical} to show that in the worst case, either~$\peak$ lies on the left segment~$\turnsegment$ of~$\projectedpath$ and $\peak$ coincides with~$\turningpoint$, or that we can rotate~$r$ to be vertical. In the former case we use Lemma~\ref{lem:1D-move-peak}, while in the latter case we apply Lemma~\ref{lem:vertical-r}. 
\end{proof}

Next we prove Lemmata~\ref{lem:vertical-r}-\ref{lem:1D-peak-or-vertical}, to prove Lemma~\ref{lem:1D-peak-and-vertical}. Let~$(T, r(s_r,d_r))$ be a worst-case instance where $\turningpoint$ is not a turning point of $\projectedpath$ and let $\peak$ be the last peak on $\strategypath$ before $\turningpoint$. By Lemmata~\ref{lem:1D-move-r} and~\ref{lem:flat-rays}, $r$ is steep and infinitesimally close to~$\turningpoint$. 

\begin{lemma}\label{lem:vertical-r}
    If the peak~$\peak$ lies on left segment~$\turnsegment$ of~$\projectedpath$, and hence coincides with~$\turningpoint$, then~$r$ is vertical.
\end{lemma}
\begin{proof}
    Assume for contradiction that $\peak$ coincides with~$\turningpoint$ and that $r$ is not vertical.
    We distinguish between two cases: either the line segment through $p_0$ perpendicular to $r$ passes above~$\turningpoint$, or not.
    In the former case, we construct the terrain~$T'$ from~$T$ by moving~$\peak$ leftwards along~$\turnsegment$, until we are in the latter case. 
    This does not affect~$c^*(r)$.
    In the latter case, we rotate~$r$ around~$\turningpoint$ to become more vertical, resulting in a higher competitive ratio: $S_T$ becomes smaller and $\tau^*(r)$ becomes larger. This contradicts that~$(T,r)$ is worst case. 
\end{proof}

\begin{lemma}\label{lem:1D-move-peak}
    If $r$ is vertical, then the peak~$\peak$ lies on left segment~$\turnsegment$ of~$\projectedpath$, and hence coincides with~$\turningpoint$, for $s< \frac{4}{9}$.
\end{lemma}
\begin{proof}
    Assume for contradiction that $r$ is vertical and that $\peak$ does not lie on left segment~$\turnsegment$.
    The height value of $T$ at the $x$-coordinate of $\peak$ can be increased towards $\projectedpath$, so that $\peak$ will lie slightly higher. This changes $\strategypath$, as the turning point~$\turningpoint$ moves to the left by some arbitrarily small distance $d$ (see Figure~\ref{fig:turning-point-move}). By Lemma~\ref{lem:strip-ratio}, in the worst case $r$ also moves to the left by distance $d$. Due to the slope $s$ of $\projectedpath$, $\peak$ must have been moved up by a distance of $2ds$. This means $|S_T|$ decreases by at least $d$, due to $r$ moving to the left and being vertical, and increases by at most $2ds$, due to $\peak$ moving up: in total $|S_T|$ decreases by at least $d(1-2s)>0$, for $s < 1/2$.
    
    On the other hand, $\tau^*(r)$ also decreases. With $r$ moving $d$ towards $\peak$, $\tau^*(r)$ decreases by $d \sqrt{1+s^2}$.
    We now consider the ratio $\frac{\text{d} \tau^*(r)}{\text{d} \peak} / \frac{\text{d} |S_T|}{\text{d} \peak} \leq \frac{d \sqrt{1+s^2}}{d(1-2s)}$ between the decrease of~$\tau^*(r)$ and $|S_T|$. As this ratio is below the lower bound of $9 \sqrt{1+s^2}$ of Lemma~\ref{lem:lowerboundS} for $s<4/9$, Lemma~\ref{lem:growth-bounds} implies that slightly moving~$\peak$ towards~$\turnsegment$ increases $c^* (r)$, contradicting that~$(T,r)$ is a worst-case instance.
\end{proof}

\begin{lemma}\label{lem:1D-peak-or-vertical}
    At least one of the following holds: the peak~$\peak$ lies on left segment~$\turnsegment$ of~$\projectedpath$, and hence coincides with~$\turningpoint$, or~$r$ is vertical.
\end{lemma}
\begin{proof}
    Assume for contradiction that neither of the two properties holds. For now assume that $S_T$ is routed over~$\peak$ and let~$S_\peak$ be the line perpendicular to~$r$ through~$\peak$. We make a case distinction on whether $S_\peak$ intersects~$r$ above~$\turningpoint$ or not. We first consider the case where~$S_\peak$ hits~$r$ above~$\turningpoint$ (see Figure~\ref{fig:1D-move-peak-P}). Let~$p_\text{opt}$ be the point where~$S_T$ hits~$r$, and let~$p_\peak$ be the vertex before~$p_\text{opt}$ on the geodesic~$S_T$, coinciding with~$\peak$. Let~$q$ be the vertex on~$S_T$ before~$p_\peak$ (possibly~$q=p_0$, as in Figure~\ref{fig:1D-move-peak-P}). Consider the line segment~$qp_\text{opt}$. Because~$\peak$ does not lie on~$\turnsegment$, $qp_\text{opt}$ intersects~$P$ between~$\peak$ and~$\turningpoint$. Let~$p'_\peak$ be the intersection point, and let~$p'_\text{opt}$ be the point on~$r$ hit by the perpendicular on~$r$ through~$p'_\peak$. Finally, let~$S_T(q)$ be the geodesic~$S_T$ from~$p_0$ to~$q$. By the above,
    \begin{equation*}
        |S_T| > |S_T(q)| + |qp_\text{opt}| > |S_T(q)| + |qp'_\peak| + |p'_\peak p'_\text{opt}|.
    \end{equation*}
    
    Consider the terrain~$T'$ where, compared to~$T$, $\peak$ moved rightward along~$P$ until it coincides with~$p'_\peak$ (see Figure~\ref{fig:1D-move-peak-P}). For $T'$ we know that $|S_{T'}| = |S_T(q)| + |p'_\peak q| + |p'_\text{opt}p'_\peak| < |S_T|$. Additionally, $\tau^*(r)$ is unaffected.
    Thus, the ratio~$c^*(r)$ strictly increases, contradicting that $(T,r)$ is a worst-case instance.

    \begin{figure}[t]
    \begin{minipage}{.47\textwidth}
    \centering
    \includegraphics[page=2]{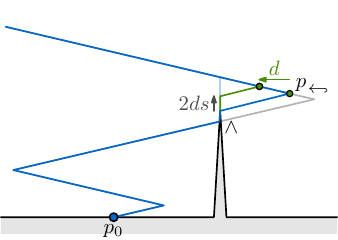}
    \caption{Moving $\peak$ upwards causes $\turningpoint$ to move left. If $\turningpoint$ moves a horizontal distance~$d$, then~$\peak$ must have moved~$2ds$. }
    \label{fig:turning-point-move}
    \end{minipage}
    \begin{minipage}{.47\textwidth}
    \vspace{5pt}
    \centering
        \includegraphics[page=2]{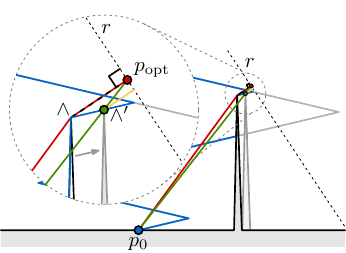}
        \caption{By moving local maximum~$\peak$ to the green point~$\peak'$, $|S_T|$ (red) strictly decreases, concatenating the green path from~$p_0=q$ to $\peak'$ and the yellow path.}
        \label{fig:1D-move-peak-P}
    \end{minipage}
    \hspace{.05\textwidth}
\end{figure}

    Notice that, when $\peak$ is not part of $S_T$, then $S_T$ hits~$r$ above~$\turningpoint$.
    In this case, the above modification to~$T'$ does not affect $\tau^*(r)$ and $|S_{T'}|=|S_T|$. Now Lemma~\ref{lem:vertical-r} applies, contradicting that~$(T,r)$ is a worst-case instance.

    Finally, consider the case where~$S_\peak$ hits~$r$ below or on~$\turningpoint$. When we rotate~$r$ around~$\turningpoint$ to become more vertical, $S_T$ decreases and $\tau^*(r)$ increases. This results in a strictly higher ratio~$c^*(r)$, contradicting that~$(T,r)$ is a worst case.
\end{proof}

\mypar{Unobstructed search path.}
Next we consider all steep visibility rays in the case that $\turningpoint$ is a turning point of $\projectedpath$, and show the following.

\begin{lemma}\label{lem:1D-unobstructed}
    In a worst-case instance~$(T, r(s_r,d_r))$, where~$\turningpoint$ is a turning point of $\projectedpath$, if $\frac{1}{9} < s \leq 1$, then~$r$ is vertical.
\end{lemma}
\begin{proof}
    Assume for contradiction that $r$ is not vertical. 
    If $S_T$ intersects~$r$ below~$\turningpoint$, rotating $r$ around $\turningpoint$ to become more vertical results in a strictly higher value $c^*(r)$, as $|S_T|$ becomes smaller and $\tau^*(r)$ becomes larger, contradicting that~$(T,r)$ is a worst case. 
    If $S_T$ hits~$r$ above~\turningpoint, then this case is equivalent to having a peak~$\peak$ at exactly $\turningpoint$, because~$\peak$ does not interfere with $S_T$. By Lemma~\ref{lem:vertical-r}, $r$ is then vertical in the worst-case.
\end{proof}

\mypar{Bounding the competitive ratio.} To finish our analysis, we combine the previous lemmata, and choose~$s$ to minimize the competitive ratio across all cases. To obtain a strategy that is feasible in practice, we assume that $|S_T| \geq 1$. That is, we do not use infinitesimally small steps to start in practice. We then adapt our strategy by first moving upwards at most one, up to the final time that $\strategypath$ is intersected, and then start following along $\strategypath$. This only shortens the search path, so the competitive ratio holds for this adjusted path as well.

\begin{restatable}{theorem}{oneDcompratio}\label{thm:1D-comp-ratio}
    Our searching strategy for searching in a 1.5D terrain achieves a competitive ratio of $3\sqrt{19/2}$ for $s=\sqrt{2}/6$.
\end{restatable}

\begin{proof}
    We consider all visibility rays~$r(s_r,d_r)$, to find a combination of~$s_r$ and~$d_r$ that maximizes $c^*(r)$. By Lemma~\ref{lem:flat-rays} we know that any visibility ray with $s_r > -s$  or where $p_t$ lies on a left \subpath of~$\strategypath$ results in a $c^*(r)$ below the lower bound of Lemma~\ref{lem:lowerboundS}. We thus consider the case where $s_r \leq -s$ and $p_t$ is on a right \subpath, both when $\projectedpath$ is intersected by~$T$ and when $\strategypath$ is unobstructed. For both of these cases, we derive a bound on $c^*(r)$ dependent on~$s$, and then choose a value of $s$ that minimizes the largest bound. 

    \mypar{Case 1: obstructed search path.} We know by Lemma~\ref{lem:1D-peak-and-vertical} that in the worst case the visibility ray~$r(s_r,d_r)$ lies just behind a local maximum~$\peak$ of~$T$, and intersects the $i$-th right \subpath of~$P$. Observe that this means that $0\leq d_r \leq 2^{i-3}$. Additionally, $\peak$ touches the ($i-1$)-th left segment of~$\projectedpath$, and $s_r=-\infty$. We can bound the competitive ratio as follows (see Figure~\ref{fig:peak-ratio}). 

    \begin{figure}
        \centering
        \includegraphics{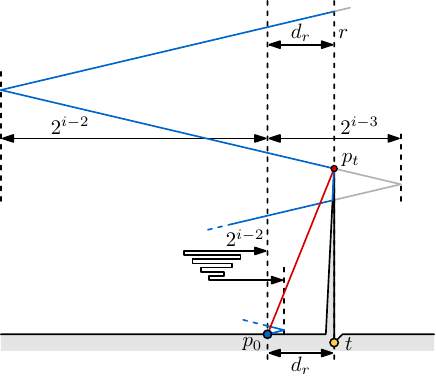}
        \caption{The construction for bounding the competitive ratio when terrain~$T$ intersects~$\projectedpath$ (in grey). The search path~$\strategypath$ (in blue) crosses visibility ray~$r$ on right segment~$h^i_r$ to enter~$\vis{\goal}$. Arrows indicate horizontal distances along~$\projectedpath$, and $S_T$ is shown in red.}
        \label{fig:peak-ratio}
    \end{figure}
    
    As~$r$ is vertical, the distance~$d_r$ coincides with the horizontal distance between~$p_0$ and~$\peak$. 
    To determine $\tau^*(r)$, consider the horizontal distance traveled along~$\projectedpath$, which is 
    \begin{equation*}
        d_r + 2\cdot 2^{i-2} + 2\cdot 2^{i-3} + \sum_{j\geq 4} 2\cdot 2^{i-j} = d_r + 8\cdot 2^{i-3}.
    \end{equation*}
    We then get~$\tau^*(r) \leq (d_r + 8\cdot 2^{i-3})\cdot \sqrt{1+s^2}$. Next, we find a lower bound on~$|S_T|$ by computing the distance between~$p_0$ and~$\peak$. For that we use the height of~$\peak$ which is
    \begin{equation*}
        h^{i-1}_\ell(0) - s\cdot d_r = s\cdot(2^{i-1}) - s\cdot d_r = s\cdot(4\cdot 2^{i-3} - d_r).
    \end{equation*}
    It follows that $|S_T| \geq \sqrt{d_r^2 +  s^2\cdot(4\cdot 2^{i-3} - d_r)^2}$. Observe that the construction of this worst case is equivalent for any odd~$i$, as both $\tau^*(r)$ and $|S_T|$ (and the upper bound on $d_r$) differ by exactly a factor $4$ between consecutive odd values of~$i$. Hence, we may choose~$i=3$, 
    such that $0\leq d_r \leq 1$ and
    \begin{align*}
        \tau^*(r) &\leq (d_r + 8)\cdot \sqrt{1+s^2}\\
        |S_T| &\geq \sqrt{d_r^2 +  s^2\cdot(4 - d_r)^2}
    \end{align*}

    Thus, when~$T$ intersects~$\projectedpath$, then $(8+d_r) \sqrt{1+s^2} / \sqrt{d_r^2 + s^2 (4 - d_r)^2}$ is an upper bound on~$c^*(r)$. Furthermore, Lemmata~\ref{lem:flat-rays} and~\ref{lem:1D-peak-and-vertical} require that $\frac{2}{9} \leq s < \frac{4}{9}$.

    \mypar{Case 2: unobstructed search path.} Lemma~\ref{lem:1D-unobstructed} tells us that in the worst case, the visibility ray~$r(s_r,d_r)$ is vertical and infinitesimally close to a turning point. Thus, $s_r = -\infty$ and, since $h^i_r$ is the lowest segment crossed by~$r$, we get that $|S_T| = d_r = 2^{i-3}$. Similar to the previous case, we have that~$\tau^*(r) = (d_r + 8\cdot 2^{i-3})\cdot \sqrt{1+s^2} = (9\cdot 2^{i-3})\cdot \sqrt{1+s^2}$. We again choose~$i=3$, resulting in $d_r = 1$ and $\tau^*(r)\leq 9\sqrt{1+ s^2}$. Thus, in this case the competitive ratio is at most $9 \sqrt{1+ s^2}$ for $\frac{1}{9} < s \leq 1$.

    To finalize the analysis, we choose~$s$ such that the value $c^*(r)$ of the case with largest ratio is minimized. We observe that $(8+d_r) \sqrt{1+s^2} / \sqrt{d_r^2 + s^2 (4 - d_r)^2}$ is decreasing in~$s$ (when $0\leq d_r\leq 1$) and $9\cdot \sqrt{1+ s^2}$ is increasing in~$s$. We can hence equate the formulas for the two cases and set $(8+d_r) \sqrt{1+s^2} / \sqrt{d_r^2 + s^2 (4 - d_r)^2} = 9 \sqrt{1+ s^2}$ to find a value for~$s$ that minimizes the competitive ratio. This equality is satisfied for $s=\sqrt{2}/6$ and $d_r=4/13$. This choice of~$s$ satisfies all of the bounds on~$s$ and for this~$s=\sqrt{2}/6$ this choice of~$d_r$ maximizes~$(8+d_r) \sqrt{1+s^2} / \sqrt{d_r^2 + s^2 (4 - d_r)^2}$ --- a worst-case ray is~$r(-\infty,4/13)$. As $c^*(r)$ is an upper bounded on the competitive ratio $c(r)$ of our strategy, our strategy has competitive ratio of at most $3\sqrt{19/2}$ using $s = \sqrt{2}/6$.
\end{proof}

\section{Competitive searching on 2.5D terrains}\label{sec:2.5D}
In this section we study the searching problem in an environment that is defined by a $2.5$D terrain, which is represented by a function $T_2$. It is easy to see that, without putting additional restrictions on the terrain, achieving a bounded competitive ratio will be impossible: consider a flat terrain with arbitrarily many small pits in the terrain that are arbitrarily steep. Any searching strategy would have to move to the location of each pit in the $xy$-plane in order to look at the bottom of the pit. As we can place arbitrarily many pits within a small bounded distance from the starting point, and the target may be in any of the pits, the competitive ratio of any searching strategy would always be unbounded. We make this argument more concrete in the lower bound construction below. To restrict the set of $2.5$D terrains under consideration, we require that the maximum slope of the terrain, which corresponds to the \emph{Lipschitz constant} $\lambda$ of $T_2$, is bounded. A strategy of moving upwards from $p_0$ results in a competitive ratio of $O(\lambda)$. In the remainder of this section we show that we can achieve a competitive ratio of $O(\sqrt{\lambda})$, which matching the lower bound for $2.5$D terrains.      

\mypar{Lower bound.}
We first show a lower bound on the competitive ratio for any searching strategy on $2.5$D terrains. Since this lower bound is a function of $\lambda$, this directly implies that the competitive ratio is unbounded if we do not limit the maximum slope of the terrain. 

\begin{figure}[b]
    \centering
    \includegraphics{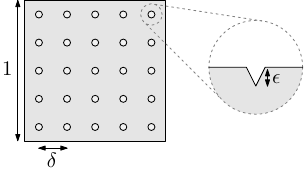}
    \caption{The terrain construction for the lower bound on~$2.5$D terrains.}
    \label{fig:2D-LB}
\end{figure}

\begin{theorem}
    The competitive ratio for searching on $2.5$D terrains with maximum slope $\lambda$ is at least $\Omega(\sqrt{\lambda})$.
\end{theorem}
\begin{proof} 
Consider a flat terrain containing a regular grid of $k \times k$ pits formed by a cone of maximum slope $\lambda$, where $\delta = 1/k$ is the distance between the centers of two adjacent pits, and $\epsilon > 0$ is the depth of each pit, where $\epsilon$ and $k$ will be chosen later (see Figure~\ref{fig:2D-LB}). 
For convenience we assume that the starting point of the searching problem is exactly a distance~$\delta$ to the left from the lower-left pit in the grid at height $0$. Now consider a searching strategy for this terrain, represented by a path $P$. 

First assume that the maximum height $z$ that $P$ reaches before being able to see the bottom of the last pit satisfies $z > \frac{\sqrt{\lambda}}{4}$. Then $P$ must travel at least a distance $z > \frac{\sqrt{\lambda}}{4}$ before seeing the last pit. The minimum travel distance to see this pit is less than $\sqrt{2}$. Hence, the competitive ratio is at least $\frac{\sqrt{\lambda}}{4\sqrt{2}} = \Omega(\sqrt{\lambda})$.

Now assume that $P$ stays under the height of $z = \frac{\sqrt{\lambda}}{4}$. By construction of the pits, this implies that the searcher must be within a horizontal distance of $(z + \epsilon)/\lambda$ from the center of the pit to see the bottom of the pit (this is the radius of the cone of a pit when extended to height $z$). As such, after checking one pit, the searcher must travel at least a distance $\delta - 2 (z + \epsilon)/\lambda$, which is the distance between two cones at height $z$, before being able to check another pit. If we choose $\epsilon = \frac{\sqrt{\lambda}}{4}$, then this distance is at least $\delta - \frac{1}{\sqrt{\lambda}}$. The total (horizontal) distance traveled by $P$ before seeing the last pit is then at least $k^2 (\delta - \frac{1}{\sqrt{\lambda}}) = k - \frac{k^2}{\sqrt{\lambda}}$, as there are $k^2$ pits in total. By choosing $k =  \frac{\sqrt{\lambda}}{2}$ this total distance is at least $\frac{\sqrt{\lambda}}{2} - \frac{\sqrt{\lambda}}{4} = \frac{\sqrt{\lambda}}{4}$. Since the minimum travel distance to see the last pit is again less than $\sqrt{2}$, the competitive ratio of $P$ is at least $\frac{\sqrt{\lambda}}{4 \sqrt{2}} = \Omega(\sqrt{\lambda})$. \qed  
\end{proof}

\mypar{Searching strategy.}
We now present a searching strategy for $2.5$D terrains with a maximum slope $\lambda$. The aim is to match the dependency on $\lambda$ that is shown in the lower bound. We will use the prior known value of $\lambda$ to determine our search path $P_\lambda$. To simplify the analysis, our searching strategy consists of separate vertical and horizontal movement phases, explained in detail below.

In the description of our search strategy, we again make use of arbitrarily small steps at the start to simplify the analysis. When a minimum value on the length of the optimal search path is given, all bounds still hold when we simply move upwards up to this value and then continuing on the described search path. Overall, our searching path $P_\lambda$ works as follows: first, we move vertically up by a distance $\varepsilon \sqrt{\lambda}$, for some arbitrarily small value $\varepsilon > 0$. Next, we construct a square horizontal grid $G_\varepsilon$ with total length $2 \varepsilon$ centered (horizontally) around the starting point. This grid will consist of $(2 k + 1) \times (2 k + 1)$ grid cells, where $k$ is chosen large enough such that the side length of a single grid cell is at most $\frac{\varepsilon}{2 \sqrt{2 \lambda}}$. Specifically, let $k$ be the smallest integer such that $2 k + 1 \geq 4 \sqrt{2 \lambda}$. We perform a horizontal search through this grid, described in detail below, and return to the center of the grid. We then move vertically up again until we are at a height that is $2 \varepsilon \sqrt{\lambda}$ above the previous grid. Here we perform a horizontal search on a grid $G_{2\varepsilon}$ with total length $4 \varepsilon$, but where the number of grid cells is still $(2 k + 1) \times (2 k + 1)$. We then repeat this process, each time doubling the vertical distance between grids and doubling the total length of the grid, but keeping the number of grid cells the same (see Figure~\ref{fig:2D-grid}). Note that a grid $G_x$ for some $x \geq \varepsilon$ is at height $(2 x - \varepsilon) \sqrt{\lambda}$ by construction. Since we assume that $\varepsilon$ is arbitrarily small, we will simply say that $G_x$ is at height $2 x \sqrt{\lambda}$. 

To perform a horizontal search in a grid $G_x$ for some $x > 0$, we first consider the height of the terrain within the grid cells. We say a grid cell $\sigma$ is \emph{eligible} if at least one point inside $\sigma$ has a height at most the height of $G_x$ (which is $2 x \sqrt{\lambda}$). We consider the connected set of eligible cells $\connectedcomponent$ that includes cell $\sigma_0$ containing the starting point (note that $\sigma_0$ is always eligible), where two eligible cells are connected if they share a side. To perform the horizontal search in $G_x$ we construct a tour that starts in the center of $\sigma_0$, visits all the centers of cells in $\connectedcomponent$, is completely contained within the cells of $\connectedcomponent$, and eventually returns to the center of $\sigma_0$ (see Figure~\ref{fig:connected-component}). During this horizontal search, the terrain may force the searcher to increase the height, which is allowed. However, the searcher never moves back down, and hence the height will never decrease anywhere on $P_\lambda$.   

\begin{figure}[t]
    \begin{minipage}{.45\textwidth}
    \centering
    \includegraphics[page=1]{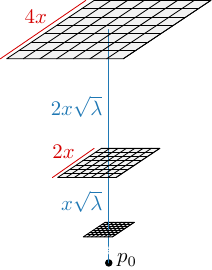}
    \caption{Three iterations with grids constructed with $2k+1$ by $2k+1$ cells.}
    \label{fig:2D-grid}
    \end{minipage}
    \hspace{.05\textwidth}
    \begin{minipage}{.45\textwidth}
    \centering
    \vspace{16pt}
    \includegraphics[page=2]{2D-Grid.pdf}
    \caption{Connected component $\connectedcomponent$ in red. Gray area is where $T$ is above the grid.}
    \label{fig:connected-component}
    \end{minipage}
\end{figure}

\mypar{Analysis.}
We first establish useful properties on the horizontal searches in grids. 
\newpage
\begin{restatable}{lemma}{gridsize}\label{lem:grid}
    Let $G_x$ be a horizontal grid used in $P_\lambda$ for some $x > 0$.
    \begin{enumerate}
        \item[\emph{(1)}] The number of grid cells in $G_x$ is $O(\lambda)$. \smallskip \\
        During a horizontal search in $G_x$:
        \item[\emph{(2)}] The amount of horizontal movement is at most $O(x \sqrt{\lambda})$.
        \item[\emph{(3)}] The amount of vertical movement is at most $\frac{x \sqrt{\lambda}}{2}$.
    \end{enumerate}
\end{restatable}

\begin{proof}
For (1) recall that the number of cells in $G_x$ is independent of $x$ and is $(2k+1)^2$ by construction, where $k$ is the smallest integer for which $2 k + 1 \geq 4 \sqrt{2 \lambda}$. This directly implies that $k = O(\sqrt{\lambda})$ and hence the number of grid cells in $G_x$ is $O(\lambda)$. 

For (2), we must bound the length of the tour that visits the centers of the reachable eligible cells in $\connectedcomponent$. For that we consider the minimum spanning tree (MST) on the centers of the grid cells in $\connectedcomponent$. The edges in this tree can only consist of edges between two neighboring grid cells that share a side. Thus, every edge in the MST has a length that corresponds to the side length of a single grid cell, which is at most $\frac{x}{2 \sqrt{2 \lambda}}$ by construction. Since the number of edges in the MST is $|\connectedcomponent|$, and $\connectedcomponent$ contains at most all cells in $G_x$, the total length of the MST is at most $O(\lambda) \times \frac{x}{2 \sqrt{2 \lambda}} = O(x \sqrt{\lambda})$ due to (1). Since the length of the optimal tour through all centers of cells in $\connectedcomponent$ is at most twice the length of the MST, the stated bound follows. 

For (3), consider any eligible cell $\sigma$ in $G_x$. By definition, there must be a point inside $\sigma$ with height at most the height of the grid. Since $\sigma$ has a side length of at most $\frac{x}{2 \sqrt{2 \lambda}}$, the maximum horizontal distance between two points in $\sigma$ is at most $\frac{x}{2 \sqrt{\lambda}}$. Given that the maximum slope of the terrain is $\lambda$, the maximum height in $\sigma$ is at most $\lambda \frac{x}{2 \sqrt{\lambda}} = \frac{x \sqrt{\lambda}}{2}$. Since the tour is contained to eligible cells by construction and the searcher can never decrease height, this is also the maximum amount of vertical movement. 
\end{proof}

Note that property (3) of Lemma~\ref{lem:grid} implies that the search path $P_\lambda$ is indeed valid, as the distance between grids $G_x$ and $G_{2 x}$ is $2 x \sqrt{\lambda}$, which is greater than $\frac{x \sqrt{\lambda}}{2}$. Thus, it is never necessary to move down again to reach the next grid in $P_\lambda$. We can now bound the length of $P_\lambda$ at a particular height along the path. 

\begin{lemma}\label{lem:length-to-grid}
The length of $P_\lambda$ up to the point of reaching a horizontal grid $G_x$ is at most $O(x \sqrt{\lambda})$.
\end{lemma}
\begin{proof}
The total amount of vertical movement in $P_\lambda$ simply corresponds to the height of $G_x$, which is $2 x \sqrt{\lambda}$ by construction. For the horizontal movement we have to consider the grids $G_{x/2}, G_{x/4}, G_{x/8}, \ldots$, which by Lemma~\ref{lem:grid} induce a horizontal movement of at most $\sum_{i=1}^\infty O(x/2^i \sqrt{\lambda}) = O(x \sqrt{\lambda})$. The stated bound follows from adding the horizontal and vertical movement in $P_\lambda$.
\end{proof}
Next, we use $\lambda$ to determine when a point on $P_\lambda$ can see the target $t$. 

\begin{lemma}
    \label{lemma:cone-to-goal}
    If $p$ is a point that can see $\goal$, then any point $p^*$ in the upwards cone starting at $p$ with slope $\lambda$ can see $\goal$.
\end{lemma}

\begin{proof} 
    Since the slope is bounded by $\lambda$, the upwards cone with slope $\lambda$ above any point that lies above the terrain must be unobstructed. Furthermore, the line segment between $p$ and $\goal$ is unobstructed. Hence, the upwards wedge with slope $\lambda$ over the path between $p$ and $\goal$ is also unobstructed (see Figure~\ref{fig:cone-to-goal}). 
    Since the cone above $p$ is unobstructed and the wedge is unobstructed, the line segment between $p^*$ and $\goal$ is unobstructed. \qed
\end{proof}

\begin{figure}
    \centering
    \includegraphics{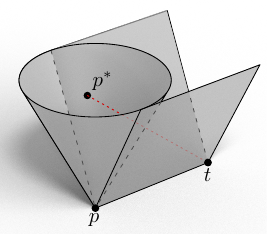}
    \caption{Illustration of Lemma~\ref{lemma:cone-to-goal}. Since the cone upwards from $p$ is unobstructed, any point $p^*$ in the cone can see $\goal$.}
    \label{fig:cone-to-goal}
\end{figure}

\begin{theorem}
Our strategy for searching in a 2.5D terrain with maximum slope~$\lambda$ achieves a competitive ratio of at most $O(\sqrt{\lambda})$.
\end{theorem}
\begin{proof}
Let $\pointopt$ be the point with the shortest distance to $p_0$ that can see $\goal$, and let $d(p_0,\pointopt)$ be the distance from $p_0$ to $\pointopt$. Furthermore, let $\vee$ be the cone cast upward from $p_0$ with slope $\frac{\sqrt{\lambda}}{2}$. For our analysis we consider two different cases: (1) $\pointopt$ lies inside of $\vee$, or (2) $\pointopt$ lies outside of $\vee$.

\mypar{Case 1: $\pointopt$ lies within $\vee$.}
Let $z$ be the height of $\pointopt$ and let $r$ be the horizontal distance from $p_0$ to $\pointopt$. Since $\pointopt$ lies within $\vee$, we know that $z \geq \frac{\sqrt{\lambda}}{2} r$. If we cast a ray directly upwards from $p_0$, we hit the cone from $\pointopt$ with slope $\lambda$ at height $z^* = z + \lambda r$. By Lemma~\ref{lemma:cone-to-goal}, we see $\goal$ from that intersection point (or any point directly above it). The next horizontal grid $G_x$ of $P_\lambda$ is at height $\leq 2 z^*$, so $2 x \sqrt{\lambda} \leq 2 z^*$ or $x \leq \frac{z^*}{\sqrt{\lambda}}$. By Lemma~\ref{lem:length-to-grid} this implies that the searcher travels at most a distance of $O(z^*)$ before seeing $t$. Since the minimum distance to reach $\pointopt$ is at least $z \geq \frac{\sqrt{\lambda}}{2} r$, we get that $z^* / z \leq 1 + 2 \sqrt{\lambda}$. Hence, the competitive ratio in this case is at most $O(z^* / z) = O(\sqrt{\lambda})$.

\mypar{Case 2: $\pointopt$ lies below $\vee$.}
Let again $z$ be the height of $\pointopt$ and let $r$ be the horizontal distance from $p_0$ to $\pointopt$. Since $\pointopt$ lies below $\vee$, we know that $z < \frac{\sqrt{\lambda}}{2} r$. Consider the first time that a cell $\sigma$ directly above $\pointopt$ is visited by $P_\lambda$ during a horizontal search of a grid $G_x$. Since $x \geq r$, the vertical distance between $\pointopt$ and $\sigma$ is at least $2 x \sqrt{\lambda} - z \geq \sqrt{\lambda} (2 x - \frac{r}{2}) \geq \frac{3 x}{2} \sqrt{\lambda}$. Hence, the upwards cone from $\pointopt$ with slope $\lambda$ intersects the horizontal plane at $G_x$ in a circle with radius $\frac{3 x}{2} \sqrt{\lambda} / \lambda = \frac{3 x}{2 \sqrt{\lambda}}$. Since the side length of $\sigma$ is at most $\frac{x}{2 \sqrt{2 \lambda}}$, this circle also contains the center of $\sigma$, from which we see $t$ due to Lemma~\ref{lemma:cone-to-goal}. Thus, the target is found at the latest during the horizontal search on $G_x$.  Lemmata~\ref{lem:length-to-grid} and~\ref{lem:grid} (property 2 and 3) then imply that we travel at most a distance of $O(x \sqrt{\lambda})$ before we find~$t$. 

We now consider the distance $d(p_0,\pointopt)$. By construction, the horizontal search on grid $G_{x/2}$ did not visit a cell above $\pointopt$. We consider two possible cases. If the grid $G_{x/2}$ does not contain any cell directly above $\pointopt$, then $r > x/2$. In that case $d(p_0,\pointopt) > x/2$ and hence we obtain a competitive ratio of $O(x \sqrt{\lambda}) / (x/2) = O(\sqrt{\lambda})$. If $G_{x/2}$ does contain a cell $\sigma'$ directly above $\pointopt$, then $\sigma'$ was not part of $\connectedcomponent$ for $G_{x/2}$. But then, in order to reach the point $\pointopt$ from $p_0$, we must either reach a height of $x \sqrt{\lambda}$ (the height of $G_{x/2}$), or we must leave the horizontal domain of $G_{x/2}$. In both cases the shortest distance from $p_0$ to $\pointopt$ is at least $x/2$ (or even $x \sqrt{\lambda}$ in the first case). Thus, we again obtain a competitive ratio of $O(\sqrt{\lambda})$.
\end{proof}

\section{Conclusion}\label{sec:conclusion}
The lower and upper bound for $1.5$D terrain might be improved with a more intricate example and more extensive analysis respectively. 
For our search strategies we assumed that the terrain is given beforehand. However, our searching strategy for $1.5$D terrains is affected by the terrain only when obstructed, thus the searching strategy can handle unknown terrains. This does not hold for our strategy on $2.5$D terrains. Though we can address terrain on the fly, we crucially use the maximum slope $\lambda$ to construct our search path. It would be interesting to study whether an efficient strategy exists that does not require~$\lambda$ to be known. Another direction for future research is to extend the result on $2.5$D terrains to special types of polyhedral domains, such as star-shaped polyhedra. An important question here is how to redefine the parameter $\lambda$ for polyhedral domains such that the competitive ratio can be bounded in terms of that parameter. 

\bibliography{bibliography.bib}

\end{document}